\documentclass{article}

\usepackage{arxiv}

\usepackage[utf8]{inputenc} 
\usepackage[T1]{fontenc}    
\usepackage{hyperref}       
\usepackage{url}            
\usepackage{booktabs}       
\usepackage{amsfonts}       
\usepackage{nicefrac}       
\usepackage{microtype}      
\usepackage{lipsum}		
\usepackage{graphicx}
\usepackage[numbers]{natbib}
\usepackage{doi}

\usepackage{enumerate}
\usepackage[shortlabels]{enumitem}

\newcommand{\skeletondef}{S_k = (V_T, E_T)}
\newcommand{\skeleton}{S_k}

\newcommand{\child}[2]{c_{#1,#2}} 
\newcommand{\leaf}[2]{l_{#1,#2}} 

\newcommand{\degreechildren}{\delta_c}

\newcommand{\height}{h}
\newcommand{\inducedGraphRecursiveCalls}{\mathbb{G}}

\newcommand{\treeDecompAHCall}[4]{\textsc{Ah-Td}$(#1,#2,#3, #4)$}
\newcommand{\treeDecompAH}{\textsc{Ah-Td}} 

\newcommand{\treeDecompHalinCall}[1]{\textsc{H-Td}$(#1)$}
\newcommand{\treeDecompHalin}{\textsc{H-Td}}

\usepackage{dutchcal}

\usepackage{comment}
\usepackage{ulem}
\usepackage{soul}

\usepackage{float}

\usepackage{multicol}

\usepackage{mathtools} 
\usepackage{amssymb}
\usepackage{amsthm}

\theoremstyle{definition}

\newtheorem{proposition}{Proposition}
\newtheorem{corollary}{Corollary}
\newtheorem{definition}{Definition}
\newtheorem{remark}{Remark}
\newtheorem{notation}{Notation}

\usepackage{graphicx}
\usepackage{subcaption}

\title{A Practical Linear Time Algorithm for Optimal Tree Decomposition of Halin Graphs}

\author{ \href{https://orcid.org/0009-0007-8014-204X}{\includegraphics[scale=0.06]{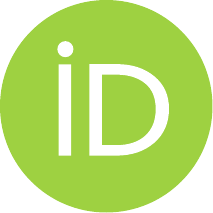}\hspace{1mm}J.A. Alejandro-Soto} \\
        Área de Computación\\
        Centro de Investigación en Matemáticas (CIMAT)\\
        Guanajuato, México\\
	\texttt{jose.alejandro@cimat.mx} \\
	\And
	\href{https://orcid.org/0000-0001-9326-7713}{\includegraphics[scale=0.06]{orcid.pdf}\hspace{1mm}Joel Antonio Trejo-Sanchez} \\
	Unidad Mérida\\
	Centro de Investigación en Matemáticas (CIMAT)\\
	Yucatán, México \\
	\texttt{joel.trejo@cimat.mx} \\
    \And
	\href{https://orcid.org/0000-0002-5431-5927}{\includegraphics[scale=0.06]{orcid.pdf}\hspace{1mm}Carlos Segura} \\
	Área de Computación\\
        Centro de Investigación en Matemáticas (CIMAT)\\
        Guanajuato, México\\
	\texttt{carlos.segura@cimat.mx} \\
}

\renewcommand{\shorttitle}{}

\hypersetup{
pdftitle={A template for the arxiv style},
pdfsubject={q-bio.NC, q-bio.QM},
pdfauthor={David S.~Hippocampus, Elias D.~Striatum},
pdfkeywords={First keyword, Second keyword, More},
}

\begin{document}
\maketitle

\begin{abstract}
This work proposes \textsc{H-Td}, a practical linear-time algorithm for computing an optimal-width tree decomposition of Halin graphs. 
Unlike state-of-the-art methods based on reduction rules or separators, \textsc{H-Td} exploits the structural properties of Halin graphs. 
Although two theoretical linear-time algorithms exist that can be applied to graphs of treewidth three, no practical implementation has been made publicly available. 
Furthermore, extending reduction-based approaches to partial $k$-trees with $k > 3$ results in increasingly complex rules that are challenging to implement. 
This motivates the exploration of alternative strategies that leverage structural insights specific to certain graph classes.
Experimental validation against the winners of the Parameterized Algorithms and Computational Experiments Challenge (PACE) 2017 and the treewidth library \texttt{libtw} demonstrates the advantage of \textsc{H-Td} when the input is known to be a Halin graph.
\end{abstract}



\section{Introduction}\label{sec-intro}

Let $G=(V,E)$ be an arbitrary undirected graph. A \textit{tree decomposition} is a structure  $\mathbb{T} = ( T, \{X_t\}_{t \in V_T})$ where certain conditions hold (see Section \ref{Sec:PrelTerm}), $T = (V_T, E_T)$ is a tree and each node $t \in V_T$ contains a \textit{bag} $X_t$ which is a subset of $V$. This decomposition allows for the  design of efficient algorithms for certain combinatorial optimization problems in the original graph $G$ \citep{lautemann1988decomposition}.  

The \textit{width} of a tree decomposition of a graph $G$ is the size of the largest bag minus one in such a tree decomposition. The \textit{treewidth} of $G$ is the minimum width among all possible tree decompositions of $G$. Computing the treewidth of an arbitrary graph is NP-hard \citep{arnborg1987complexity}. 
For a fixed integer $k$ there is a linear time algorithm that determines if a graph has treewidth at most $k$ and if so, it finds a tree decomposition of width at most $k$~\citep{bodlaender1993linear}. 
However, this algorithm has been studied in depth and is not practical, since it has a large hidden constant that can be up to $2^{O(k^3)}$~\citep{roehrig1998tree}. 
An alternative approach is to compute tree decompositions for partial $k$-trees. 
Two linear-time algorithms were given, one for $k \leq 3$~\citep{LinearTDWidth3} without available implementation and one for $k = 4$~\citep{TW4-Rules} which was tested in a later work~\citep{TW4-Implementation}. 
However, increasing $k$ from three to four required the use of an infinite set of safe reduction rules, together with $60$ leaf structures to be complete. 
Some authors believe that further extensions of this approach for $k \geq 5$ will not be rewarding~\citep{TW4-Implementation}. 
These facts motivates the study of this work: an alternative approach for computing tree decompositions based on the structural properties of the graph classes.

In this paper, we propose \textsc{H-Td}, a linear-time algorithm without large hidden constants to compute an optimal tree decomposition of an arbitrary Halin graph. First, we define two special subgraphs of a Halin graph: the Almost Halin type 1 ($AH_1$) and the Almost Halin type 2 ($AH_2$). Next, we present a linear-time algorithm to compute a tree decomposition for these subgraphs. Finally, we introduce a linear-time algorithm to compute an optimal tree decomposition for arbitrary Halin graphs. In the process of computing such a decomposition for Halin graphs, the tree decompositions algorithms for $AH_1$ and $AH_2$  are required.

The rest of the paper is organized as follows. Section \ref{sec:relwork} summarizes the most relevant work. Section \ref{Sec:PrelTerm} presents some terminology that is useful for describing the algorithm. Section \ref{sec:AdjacencyOrder} gives the order of the adjacency lists needed for the algorithms. Section \ref{sec:Almost} defines the Almost Halin graphs and describes the algorithm for computing their tree decompositions. Section \ref{section:HalinPropuesta} presents \textsc{H-Td}, the algorithm that computes a tree decomposition of a Halin graph. Finally, Section~\ref{section:ExperimentalValidation} validates the performance of \textsc{H-Td} and Section~\ref{Section:Concluding} presents the conclusions of this work. 


\section{Related work}
\label{sec:relwork}

Certain NP-hard problems on graphs admit polynomial-time solutions for specific graph classes. For instance, in trees and cacti, there are algorithms that provide optimal solutions for some NP-hard problems \citep{mjeldek, baiou2014algorithms, flores2018algorithm}. Moreover, graphs with a constant-bounded or known small treewidth exhibit important properties for certain NP-hard optimization problems \citep{lampis2012algorithmic, KosterAplication}. Specifically, for some NP-hard optimization problems, it is well known that these problems admit polynomial-time optimal solutions if the treewidth of $G$ is bounded \citep{downey2013fundamentals}.

Deciding whether the treewidth of an arbitrary graph is bounded by a given integer $k$ is NP-complete \citep{arnborg1987complexity}. However, there exists some approximation algorithms dealing with computing the tree decomposition of a given graph. Feige \textit{et al.} \citep{feige2005improved} designed a polynomial-time algorithm that computes a tree decomposition of width $O(k \sqrt{ \log k })$ in a graph with $n$ vertices and treewidth $k$. Kammer \citep{kammer2016approximate} developed an  algorithm that in a planar graph of treewidth $k$ with $n$ vertices, computes a tree decomposition of width $O(k)$ in $O(n k^2 \log k)$ time.

Although the problem of determining the complexity class of computing the treewidth of planar graphs is open \citep{dissaux2023treelength}, algorithms exist to compute almost optimal tree decompositions on certain planar graphs. In particular, Katsikarelis proposed an algorithm to compute the tree decomposition of $k$-outerplanar graphs of width at most $3k-1$ \citep{katsikarelis2013computing}. Since Halin graphs are $2$-outerplanar graphs, the Katsikarelis algorithm ensures a tree decomposition of width five in these graphs. However, the treewidth of a Halin graph is three \citep{bodlaender1988planar}.

Despite computing the treewidth of a graph is NP-hard in general, for specific classes of graphs it is possible to determine a bound for the treewidth and compute a tree decomposition in polynomial time, as in the cases of cographs \citep{bodlaender1993pathwidth} and circular graphs \citep{sundaram1994treewidth}. Bodlaender and Kloks developed polynomial-time algorithms to compute a tree decomposition in graphs with bounded treewidth \citep{bodlaender1991better}. Later, the complexity was improved to be linear~\citep{bodlaender1993linear}. However, this algorithm is impractical due to the large hidden constants in their complexity \citep{niedermeier2006invitation, kammer2010treelike}. The constant factor has been shown to be up to $2^{O(k^3)}$~\citep{roehrig1998tree}.

Arnborg~\citep{ArnborgPartial3Trees} proposed a set of reduction rules to recognize partial $3$-trees in $O(N^3)$ time and pointed out that it may be accelerated to $O(N \log N)$. 
Later, with one more reduction rule, it was shown how to compute in linear time a tree decomposition of width at most three~\citep{LinearTDWidth3}. 
However, to the best of our knowledge there is no code available. 
A similar approach to recognize graphs with treewidth at most four resulted in an infinite set of safe reductions together with a set of $60$ leaf structures to be complete~\citep{TW4-Rules}. 
The author showed how to derive a linear-time algorithm from the reductions and it was tested in a later work, confirming the expected execution times for recognizing random partial $4$-trees~\citep{TW4-Implementation}. 
There are $75$ minimal forbidden minors for treewidth at most four~\citep{ThesisSanders}, which in some sense justify the increasing number of reductions rules of~\citep{TW4-Rules}. 
This indicates that for larger values of treewidth, algorithms based on reduction rules will require increasingly complex procedures~\citep{TW4-Implementation}. 
Therefore, alternative methods should be explored, such as exploiting the structure of graph classes, as is done in this work.

During the first two editions of the Parameterized Algorithms and Computational Experiments Challenge (PACE) one of the problems was to compute a tree decomposition~\citep{PACE2016, PACE2017}. 
There were two tracks, one of exact algorithms and one of heuristics. 
The solutions in the exact track in year 2017 outperformed the solutions of 2016~\citep{PACE2017}. 
In~\citep{libtw} a library \texttt{libtw} for computing the treewidth and tree decomposition of graphs is presented. 
Bodlaender et al.~\citep{PREPROCESSING_RULES} used a set of safe reduction rules motivated by properties of real-life probabilistic networks. 
However, some of the rules are implemented in quadratic time. 

We compare our algorithm with the best results of PACE 2017 and with the \texttt{libtw} library since these are the best reproducible results with available code.


\section{Preliminaries and terminology}
\label{Sec:PrelTerm}

Let $G = (V_G, E_G)$ be an undirected graph. The \textit{neighborhood} of a vertex $v \in V_G$, denoted by $N(v)$, is the set of vertices $u \in V_G$ such that $(v, u) \in E_G$. The \textit{degree} of a vertex $v$, denoted by $\delta(v)$, is the cardinality of $N(v)$.

A graph $S = (V_S, E_S)$ is an \textit{induced subgraph} of $G$ if $V_S \subseteq V_G$, $E_S \subseteq E_G$ and for any two vertices $u, v \in V_S$, $(u, v) \in E_G$ if and only if $(u, v) \in E_S$. A \textit{tree} $T=(V_T,E_T)$ is a connected graph containing no cycles. The tree is a \textit{rooted tree} if there exists a special node referred to as the root $r$. Vertices of degree one are called \textit{leaves}. In some cases, the root may also be a leaf. An \textit{internal vertex} is a vertex that is not a leaf. Each vertex $v$ different from the root has a \textit{parent} defined as the neighbor of $v$ that is closer to the root. The remaining neighbors of $v$ are its \textit{children}. The number of children of $v$ is referred to as the \textit{degree children} of $v$ and is denoted by $\degreechildren(v)$. 
Once the tree is rooted, we assume that the graph representation is given by adjacency lists. 
A vertex $w$ is an \textit{ancestor} of $v$ if the path from the root $r$ to $v$ passes through $w$. 
Note that under this definition a vertex, $v$ is considered its own ancestor.
The \textit{subtree} $T_v$ rooted at $v$ is the induced subgraph of $T$ consisting of all the vertices $u$ for which $v$ is an ancestor. 
The height of a vertex $v$, denoted by $h(v)$, is defined as the maximum distance between $v$ and any leaf within its subtree $T_v$. 

\begin{remark}
    In any tree $T=(V_T,E_T)$, there exists at least one vertex $u \in V_T$ such that all the children of $u$ are leaves. 
    \label{rem:parentleaves}
\end{remark}

Let $T=(V_T, E_T)$ be a tree embedded in the plane with no vertices of degree two such that $|V_T|\geq 4$. Let $V_C \subset V_T$ be the set of leaves of $T$. Let $E_C$ be a set of edges connecting all the vertices of $V_C$ in the cyclic order of its embedding. Then, the graph $H=(V_H=V_T, E_H=E_T \cup E_C)$ is a \textit{Halin graph} \cite{halin1971studies}, and  $C=(V_C,E_C)$ is its \textit{leaf cycle}.  The smallest Halin graph is $K_4$. 
We refer to the tree of the Halin graph as the \textit{skeleton} $\skeletondef$  \citep{fomin20063}. In contrast, Flores \textit{et. al.} define the skeleton $\skeleton$ as the subgraph of $H$ induced by $V_S=V_H \setminus V_C$ \citep{flores2020distributed}. 
Eppstein presented a linear time algorithm based on two reduction rules for recognizing Halin graphs and in the positive case determine a valid decomposition in the skeleton and the leaf cycle~\citep{Eppstein:2016}.

Some additional relevant properties of any Halin graph $H$ are the following: $H$ is $2$-outerplanar and 3-connected \citep{halin1971studies}, $H$ is Hamiltonian and 1-Hamiltonian ($H$ remains hamiltonian after deleting any vertex or any edge from $H$) \citep{cornuejols1983halin}, $H$ is almost pancyclic ($H$ contains cycles of almost every length between 3 and $\vert V_H \vert$) \citep{skowronska1985pancyclicity}, 
$H$ always has at least two size-3 cliques \citep{bondy1985lengths}. 

\begin{remark}
    \label{rem:subtree_leaves_connected}
    Let $H$ be a Halin graph with skeleton $\skeletondef$ rooted at $r$ and leaf cycle $C$. For each vertex $v \in V_T$ the leaves in the subtree $T_v$ induce a connected component in $C$.
\end{remark}

\begin{remark}
\label{remark:vertexdegree3}
In a Halin graph $H = (V, E)$ with skeleton $\skeletondef$, there exists an internal vertex $v \in V_T$ with degree at least three.
\end{remark}

Let $G=(V,E)$ be a graph. A \textit{tree decomposition} of $G$  is a pair $\mathbb{T} = ( T, \{X_t\}_{t \in V_T})$ where $T = (V_T, E_T)$ is a tree and each node $t \in V_T$ contains a \textit{bag} $X_t$ which is a subset of $V$ \citep{diestel2012graph}. Additionally, the following conditions are required.

\begin{enumerate}[start=1,label={(\bfseries C\arabic*)}] \label{conditionsTreeDecompositon}
    \item $\bigcup_{t \in V_T} X_t = V$ \label{condition:C1}
    \item For each edge $(u,v) \in E$, there exists a node $t$ of $T$ such that both $u$ and $v$ are included in the bag $X_t$. We say that an edge is included in a bag if its endpoints are included. \label{condition:C2}
    \item If a vertex $u \in V$ is included in two bags $X_A, X_B$, all the bags in the path from $A$ to $B$ contain the vertex $u$. Alternatively, the set of nodes that include $u$ in their bag induces a connected subtree of $T$. \label{condition:C3}
\end{enumerate}

We adopt the convention used in~\citep{Cygan2015ParameterizedAlgorithms} where the term \textit{node} refers to the vertices of $\mathbb{T}$ while the term \textit{vertex} is reserved for $H$.


\section{Adjacency order}
\label{sec:AdjacencyOrder}

In order to build the tree decomposition, we need to give a particular order to the adjacency list. In this section we introduce this specific order and notation derived from it that will be used in the description of the final algorithm. 
This process assumes that it can distinguish the skeleton and the leaf cycle  of the Halin. In the \textsc{H-Td} algorithm we will use Eppstein's procedure~\citep{Eppstein:2016} to get this information in linear time. 

\begin{proposition}
\label{TEO:ORDEN}
Given a Halin graph $H = (V, E)$ with skeleton $\skeletondef$ and leaf cycle $C=(V_C,E_C)$, for each vertex $r \in V_T$ with degree at least three, there exists a specific traversal of $C$ such that if we root $\skeleton$ at $r$, it is possible to rearrange the adjacent list of each vertex in $V_T$ in such a way that a depth-first search starting from $r$ visits the leaves in the same order as in this traversal of $C$.
\end{proposition}
\begin{proof}
First, let $r \in V_T$ be an internal vertex with degree at least three (as ensured by Remark \ref{remark:vertexdegree3}). Next, root $\skeleton$ at  $r$. 
Choose a leaf $v \in V_T$ such that the lowest common ancestor of $v$ and one of its neighbors in $C$, say $w$, is $r$. 
Note that such a vertex $v$ exists because each child of $r$ induces a non-empty subtree, and since the leaves form a cycle, there will be at least one pair of leaves from different subtrees that are adjacent in $C$.

For each leaf $u \in V_T$, define $x_u$ as the number of leaves visited before $u$ when traversing $C$ starting from $v$ and ending in $w$. Specifically, we have $x_v=0$ and $x_w=|V_C|-1$.  
For each internal vertex $y \in V_T$, define $x_y$ as the minimum $x_u$ for any leaf $u$ in the subtree of $y$. 
See Figure \ref{fig:AdjOrdering} for an example of these values. 

For each internal vertex, sort its children in increasing order based on their $x$-values. We will now demonstrate that this ordering achieves the desired structure.

Since $v$ and $w$ were specifically chosen such that their lowest common ancestor is $r$, the only vertex that can have both of these leaves in its subtree is $r$. 
Furthermore, by Remark \ref{rem:subtree_leaves_connected} and since the $x$-values for the leaves were assigned incrementally in a specific direction, we obtain that for each vertex $y \in V_T$, the $x$-values of the leaves within its subtree form a continuous interval $[a_y, b_y]$ with $0 \leq a_y \leq b_y \leq |V_C| - 1$.
Let $z \in V_T$ be any vertex with interval $[a_z, b_z] \subseteq [0, |V_C|-1]$. By the same argument, if $y_1, y_2, ..., y_k$ are the children of $z$ (arranged in the proposed order) with intervals $[a_{y_1}, b_{y_1}], [a_{y_2}, b_{y_2}], ..., [a_{y_k}, b_{y_k}]$, then we have $a_{y_1} = a_z$, $b_{y_k} = b_z$ and for $1 \leq i < k$, $b_{y_i} + 1 = a_{y_{i+1}}$. Therefore, a depth-first search conducted from $r$ in this order will visit the leaves incrementally by their $x$-values, which by construction corresponds to a traversal of $C$. 
\end{proof}

\begin{figure}[!t]
  \centering 
  \includegraphics[width=0.4\textwidth]{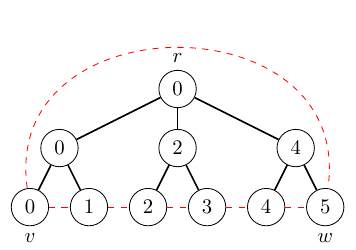}
  \caption{A Halin graph illustrating a possible selection of vertices $r$, $v$ and $w$. Solid edges represent the skeleton, while dashed edges form the leaf cycle. The $x$-values corresponding to Proposition~\ref{TEO:ORDEN} are displayed within each vertex}
  \label{fig:AdjOrdering}
\end{figure}

In Appendix \ref{appendix:LinearTeo1} we discuss how to compute the order in linear time. 
Through the rest of the paper, we assume that $\skeleton$ is rooted at an internal vertex of degree at least three 
and that the procedure to sort the adjacency list of each vertex (as described in the previous proof) is applied.

\begin{definition}
\label{defi:adjnotation}
For a vertex $v \in \skeleton$ and $1 \leq i \leq \degreechildren(v)$, we denote by $\child{v}{i}$ the $i$-th children of $v$.  In addition, we denote its last children by $\child{v}{\infty}$.
\end{definition}

\begin{definition}
\label{defi:leafnotation}
For a vertex $v \in \skeleton$, we denote by $\leaf{v}{i}$ the $i$-th leaf visited in the subtree of $v$ when running a depth-first search restricted to its subtree. In addition, we denote its last leaf visited by $\leaf{v}{\infty}$.
\end{definition}

\begin{definition}
\label{defi:leaveslabel}
For each leaf $v \in \skeleton$, we define its $x$-value to be $x_v$ as in Proposition~\ref{TEO:ORDEN}.
\end{definition}

Note that considering only the leaves, the $x$-values are unique. Furthermore, the $x$-values are numbered in the direction of a traversal of $C$ so we can define the next and previous leaf of each.

\begin{definition}
\label{defi:leafNxtPrv}
For a leaf $v \in \skeleton$ with $x$-value $x_v$, we denote by $v^+$ the leaf with $x$-value $(x_v + 1) \mod |V_C|$ and with $v^-$ the leaf with $x$-value $(x_v - 1) \mod |V_C|$. 
\end{definition}


\section{Almost Halins}
\label{sec:Almost}

When building the tree decomposition of the Halin graph, we will encounter subgraphs with two 
specific structures, referred to as Almost Halin. This section presents Almost Halin graphs and outlines a method to compute a tree decomposition of width three for these graphs in linear time. Our procedure uses the information inherited by the Halin graph $H$; particularly, the partition of the edges of a Halin into the skeleton $\skeleton$ and the leaf cycle $C$, the root of $\skeleton$ and the adjacency order of Section \ref{sec:AdjacencyOrder}.  

\begin{definition}
For $v \in \skeleton$ and $1 \leq i \leq j \leq \degreechildren(v)$ we denote by $AH_1(v, i, j)$ the induced subgraph of $H$ containing $v$ and the vertices in the subtrees of $\child{v}{k}$ for $i \leq k \leq j$, together with the vertex $\omega_1(v, j) = \leaf{\child{v}{j}}{\infty}^+$ and the edge connecting $\omega_1(v, j)$ to $\leaf{\child{v}{j}}{\infty}$.
We refer to the graphs obtained this way as \textit{Almost Halin type $1$} and we refer to $v$ as its root.  
\end{definition}

\begin{definition}
For $v \in \skeleton$ and $1 \leq i \leq j \leq \degreechildren(v)$, we denote by $AH_2(v, i, j)$ the induced subgraph of $H$ containing $v$ and the vertices in the subtree of $\child{v}{k}$ for $i \leq k \leq j$. We refer to the graphs obtained in this way as \textit{Almost Halin type $2$} and we refer to $v$ as its root. We also denote $\omega_2(v,j) = \leaf{ \child{v}{j} }{\infty}$. In the particular case of $AH2(r, 1, \degreechildren(r))$ we exclude the edge between the vertices $\leaf{r}{1}$ and $\omega_2(v, j)$.
\end{definition}

\textit{Almost Halin type $1$} graphs $AH_1(v, i, j)$ are built from trees in which consecutive leaves are connected, as in Halin graphs, and additionally the last leaf has an extra neighbor while \textit{Almost Halin type $2$} does not include the extra neighbor (see Figure \ref{fig:AlmostHalin}). 
We say that $AH_t(v, i, j)$ for $t \in \{1, 2\}$ is \textit{well-defined} if $1 \leq i \leq j \leq \degreechildren(v)$ and the resulting graph is different from $H$, which only occurs for $AH_1(r, 1, \degreechildren(r))$. Otherwise, we say that it is \textit{degenerate}. 

\begin{figure*}[!t]
  \centering
  \includegraphics[width=0.8\textwidth]{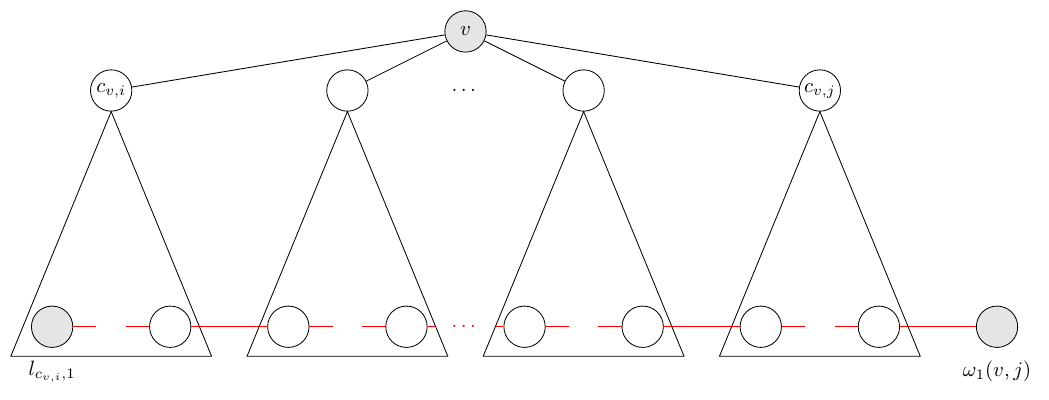}
  \caption{Illustration of $AH_1(v, i, j)$. The gray vertices are the representatives. In this case, $AH_2(v, i, j)$ is obtained by deleting the vertex $\omega_1(v,j)$}
  \label{fig:AlmostHalin}
\end{figure*}

\begin{definition}
For $v \in \skeleton$ and $1 \leq i \leq j \leq \degreechildren(v)$, we define the \textit{representatives} of a well-defined $AH_t(v, i, j)$ for $t \in \{1, 2\}$ as the set of vertices $\{v, \leaf{ \child{v}{i} }{ 1}, \omega_t(v, j)\}$.
\end{definition}

\begin{definition}
For $v, w \in \skeleton$, $1 \leq i \leq j \leq \degreechildren(v)$ and $1 \leq k \leq l \leq \degreechildren(w)$, we say that a well-defined $AH_s(w, k, l)$ is strictly smaller than a well-defined $AH_t(v, i, j)$ for $s,t \in \{1,2\}$ if the set of vertices of the former is a proper subset of the set of vertices of the latter.
\end{definition}

Now, we present \treeDecompAHCall{t}{v}{i}{j}, a linear-time algorithm to compute a tree decomposition of width three of a well-defined $AH_t(v, i,j)$ for $t \in \{1, 2\}$. 
To compute the tree decomposition of $AH_t(v, i, j)$, the algorithm recursively computes the tree decompositions of at most two strictly smaller Almost Halin.
In the general case (Figure \ref{fig:AlmostHalinTreeDecomposition}), \treeDecompAHCall{t}{v}{i}{j} calls \treeDecompAHCall{ 1 }{ v }{ i }{ i } and \treeDecompAHCall{t}{v}{i + 1}{j}. 
The reason is because their corresponding Almost Halin share only the vertex $v$ and the leaf $\leaf{ \child{v}{i + 1} }{ 1}$, which are common representatives for both graphs.

\begin{figure*}[!t]
  \centering
  \includegraphics[width=0.8\textwidth]{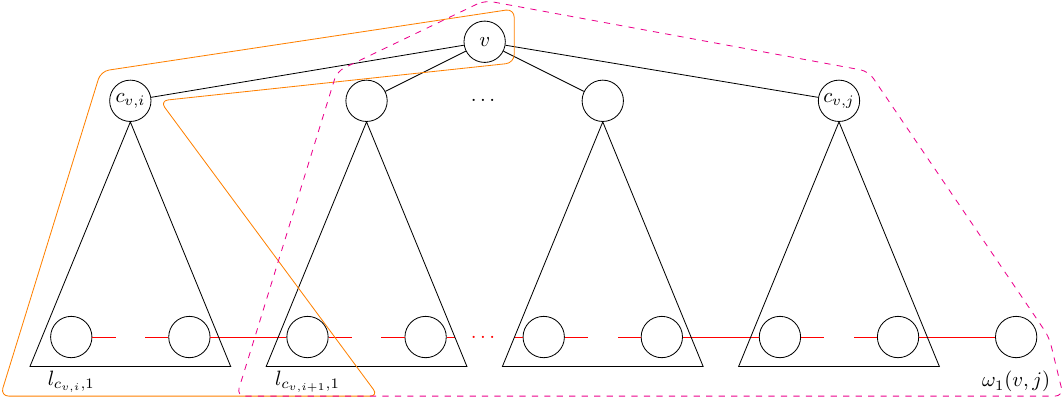}
  \caption{$AH_1(v, i, i)$ and $AH_1(v, i + 1, j)$ are enclosed in solid and dashed lines respectively}
  \label{fig:AlmostHalinTreeDecomposition}
\end{figure*}

\vspace{1em}
\noindent\textbf{\large\treeDecompAHCall{t}{v}{i}{j}}
\label{AlmostHalin1Algoritmo}

\begin{enumerate}[start=1,label={(\bfseries Step \arabic*)},leftmargin=3.5em]
    \item \label{AlmostHalin1:Step1} Create a tree $T$ with a single node $R$, whose bag is the set of representatives of $AH_t(v, i, j)$.
        \begin{itemize}
            \item If $i = j$
                \begin{itemize}[label={--}]
                    \item If $\height( \child{v}{i} ) = 0$ then go to \ref{AlmostHalin1:Step4}
                    \item Add one child $S$ to $R$ with bag $\{v, \leaf{ \child{v}{i} }{ 1}, \omega_t(v, j), \child{v}{i}\}$. Note that we include the $i$th child of $v$. 
                           \item Add one child $A$ to $S$ with bag $\{\leaf{ \child{v}{i} }{ 1}, \omega_t(v, j), \child{v}{i}\}$. 
                           \item Since $\height( \child{v}{i} ) > 0$, the graph $AH_t( \child{v}{i}, 1 ,\degreechildren(\child{v}{i}))$ is well-defined and the bag of $A$ is the set of its representatives. 
                           Append the construction of \treeDecompAHCall{t}{ \child{v}{i} }{1}{ \degreechildren(\child{v}{i}) } to $T$ identifying the node $A$ (of this call) with node $R$ (of the invoked call).
                           \item Go to \ref{AlmostHalin1:Step4}
                \end{itemize}
            \item Else ($i < j$)
                \begin{itemize}[label={--}]
                    \item Add one child $S$ to $R$ with bag $\{v, \leaf{ \child{v}{i} }{ 1}, \leaf{ \child{v}{i+1} }{ 1} , \omega_t(v, j) \}$. Note that we include the vertex $ \leaf{ \child{v}{i+1} }{ 1}$ in the bag, which is the first leaf in the subtree of the $(i+1){th}$ child of $v$; for simplicity, let us name it $x$. 
                    \item Add two children $A$ and $B$ to $S$ with bags $\{v, \leaf{ \child{v}{i} }{ 1}, x \}$ and $\{v, x, \omega_t(v, j)\}$, respectively. 
                    \item Go to \ref{AlmostHalin1:Step2}. 
                \end{itemize}
        \end{itemize}
    \item \label{AlmostHalin1:Step2}
        Since $AH_t(v, i, j)$ is well-defined, it follows that $AH_1( v, i , i)$ is also well-defined and the bag of $A$ is the set of representatives of $AH_1(v, i, i)$. Append the construction of \treeDecompAHCall{1}{v}{i}{i} to $T$, identifying the node $A$ (of this call) with node $R$ (of the invoked call). 
        Go to \ref{AlmostHalin1:Step3}.
        
    \item \label{AlmostHalin1:Step3}
        Since $AH_t(v, i, j)$ is well-defined and at this point $i < j$, it follows that $AH_t( v, i + 1 , j)$ is also well-defined and the bag of $B$ is the set of representatives of $AH_t(v, i + 1, j)$. Append the construction of \treeDecompAHCall{t}{v}{i+1}{j} to $T$, identifying the node $B$ (of this call) with node $R$ (of the invoked call). 
        Go to \ref{AlmostHalin1:Step4}.
        
    \item \label{AlmostHalin1:Step4} Return $T$
\end{enumerate}

\begin{proposition}
\label{AH-Linear}
\treeDecompAHCall{t}{v}{i}{j} finishes in linear time for any $t \in \{1, 2\}$
\end{proposition}
\begin{proof}
Since each recursive call considers at most two strictly smaller Almost Halin, the recursion will terminate. Furthermore, vertex $v$ will be considered the root of an Almost Halin in $2(j -i + 1) - 1$ calls, and each internal vertex $u$ different from $v$ will be considered the root of an Almost Halin in $2\degreechildren(u) - 1$ calls. Hence, there will be a linear amount of recursive calls and in each we make $O(1)$ operations across the three steps. Therefore, the algorithm finishes in linear time with respect to the number of vertices of $AH_t(v, i, j)$.
\end{proof}

\begin{corollary}
\label{Coro:NoCiclos}
\treeDecompAHCall{t}{v}{i}{j} does not cycle for any $t \in \{1, 2\}$
\end{corollary}

Let $\inducedGraphRecursiveCalls$ be the induced digraph of recursive calls of \treeDecompAH{}. Each vertex $\mathbcal{u}$ is a call of the algorithm and there is a directed edge $ \mathbcal{u} \to \mathbcal{v}$ when call $\mathbcal{u}$ invokes call $\mathbcal{v}$. By \hyperref[Coro:NoCiclos]{Corollary 1} $\inducedGraphRecursiveCalls$ is a directed acyclic graph.

\begin{notation}
We denote by $T(t, w, a, b)$ the set of nodes $R$, $S$, $A$, $B$ along with the edges joining them, whenever these nodes are created during the call \treeDecompAHCall{t}{w}{a}{b}.
\end{notation}

\begin{proposition}
\label{AH-TreeDecomposition}
\treeDecompAHCall{t}{v}{i}{j} computes a tree decomposition of $AH_t(v, i, j)$ of width three.
\end{proposition}
\begin{proof}
We will prove that each call to \treeDecompAH{} generates a tree decomposition of its associated Almost Halin. 
This way we will obtain that the first call computes a tree decomposition of $AH_t(v, i, j)$. 
Let us proceed by induction in a reversed topological order of $\inducedGraphRecursiveCalls$.

Note that the first vertex in the order, say $\mathbcal{u}$ associated with \treeDecompAHCall{t}{w}{a}{b}, has outdegree $0$; equivalently, 
it does not invoke any other call. Thus, the construction obtained is $T(t, w, a, b)$ and it can be explicitly 
shown that the proposition is satisfied (See Appendix \ref{appendix:BaseCase}).

Now, assume by induction that a certain prefix of vertices in the order generates a tree decomposition of its corresponding Almost Halin of width three. 
We will now show that the next vertex $\mathbcal{u}$ associated with $AH_s(w, a, b)$ (for some $s \in \{1,2\}$) also accomplishes that.
Denote by $T_{\mathbcal{u}}$ the construction obtained by $\mathbcal{u}$.
If $\mathbcal{u}$ has outdegree $0$, the result holds by the previous argument. 
If $\mathbcal{u}$ has only one dependency, the reasoning for the two dependency case can be applied with slight modifications. Thus, we assume that $\mathbcal{u}$ has outdegree $2$. 
In that case, $a < b$ and $\mathbcal{u}$ depends on the construction of $AH_1(w, a, a)$ and $AH_s(w, a + 1, b)$ (Figure \ref{fig:AlmostHalinTreeDecomposition}).
Let $\mathbcal{v}$ and $\mathbcal{w}$ be vertices of $\inducedGraphRecursiveCalls$, which represents the dependencies of $\mathbcal{u}$. 
Necessarily, both $\mathbcal{v}$ and $\mathbcal{w}$ appeared before $\mathbcal{u}$ in the order. 
By the inductive hypothesis, we assume that they generated a valid tree decomposition of its corresponding Almost Halin of width three, 
say $T_{\mathbcal{v}}$ and $T_{\mathbcal{w}}$ respectively. Since $T(s,w, a, b)$ is a tree to which we appended both $T_{\mathbcal{v}}$ and $T_{\mathbcal{w}}$ to the leaves $A$ and $B$, we find that $T_{\mathbcal{u}}$ is also a tree.

Given that the Almost Halin of both $T_{\mathbcal{v}}$ and $T_{\mathbcal{w}}$ are strictly smaller than $AH_s(w, a, b)$, all vertices included in their bags are also vertices of $AH_s(w, a, b)$. Additionally, by construction, each vertex of $AH_s(w, a, b)$ is included in at least one of $AH_1(w, a, a)$ or $AH_s(w, a + 1, b)$. Hence, $T_{\mathbcal{u}}$ satisfies \ref{condition:C1}.

Similarly, all the edges are present in one of the dependencies.
Therefore, $T_{\mathbcal{u}}$ satisfies \ref{condition:C2}.

For the last needed condition \ref{condition:C3}, consider a vertex $z$ of $AH_s(w, a, b)$. We have the following cases:
\begin{itemize}
    \item $z$ only appears in $T(s,w, a, b)$. 
        We can explicitly state that \ref{condition:C3} holds.
    \item $z$ appears in one of its dependencies and not in $T(s,w, a, b)$. 
        It will only appear in the corresponding dependency, so \ref{condition:C3} is satisfied by assumption.
    \item $z$ appears in one of its dependencies as well as in $T(s,w, a, b)$. Notice that by construction of $T(s,w, a, b)$, $z$ needs to be one of the representatives of the corresponding dependency. This means that $z$ will appear in the bag of either $A$ or $B$. Thus, the component that includes $z$ is connected.
    \item $z$ appears in both dependencies. Since $z$ is a common vertex between $AH_1(w, a, a)$ and $AH_s(w, a + 1, b)$, it follows that $z$ is either $v$ or $x$, which are representatives of the two subgraphs. By construction, both $v$ and $x$ appear in the bags of both $A$ and $B$. Thus, the component that includes $z$ is connected.
\end{itemize}
In any case, the component that includes $z$ is connected. Therefore, $T_{\mathbcal{u}}$ satisfies \ref{condition:C3}.

Thus, by the principle of mathematical induction, we conclude that each recursive call generates a tree decomposition of its associated Almost Halin. 
Therefore, at the end, we obtain a tree decomposition of $AH_t(v, i, j)$. Since each bag contains at most four vertices, the resulting tree decomposition has width three.
\end{proof}


\section{Halin graph tree decomposition}
\label{section:HalinPropuesta}

This section presents \treeDecompHalinCall{H}, a linear time algorithm that computes a tree decomposition 
of width three for a Halin graph $H$ with skeleton $\skeleton$ and leaf cycle $C$. 
Based on the previous results the function \treeDecompHalinCall{H} is concise and consists of computing the adjacency order and making a single call to \treeDecompAH, as follows.

\section*{\treeDecompHalinCall{H}} \label{HTD-algorithm}
\begin{enumerate}[start=1,label={(\bfseries Step \arabic*)},leftmargin=3.5em]
    \item Preprocess $H$ using Eppstein's algorithm~\citep{Eppstein:2016} to internally obtain the skeleton $\skeleton$ and the leaf cycle $C$
    \item Root $\skeleton$ at an internal vertex $r$ with degree at least $3$ and compute the adjacency order as described in Section \ref{sec:AdjacencyOrder} \label{Init:Step1}
    \item Return \treeDecompAHCall{ 2 }{ r }{ 1 }{ \degreechildren(r) }
\end{enumerate}

\begin{proposition}
\treeDecompHalinCall{H} computes a tree decomposition of width three for a Halin graph $H$ in linear time
\end{proposition}
\begin{proof}
By the Propositions~\ref{AH-Linear} and~\ref{AH-TreeDecomposition} it follows that \treeDecompAHCall{ 2 }{ r }{ 1 }{ \degreechildren(r) } computes a tree decomposition of width three for $AH_2(r, 1, \degreechildren(r))$ in linear time. Thus, the only possible edge that may not be included in any bag is the one connecting $\leaf{r}{1}$ with $\leaf{r}{\infty}$. However, by definition of the representatives of $AH_2(r, 1, \degreechildren(r))$, these two vertices are included in the bag of the node created in \ref{AlmostHalin1:Step1} and the result follows. 
\end{proof}


\section{Experimental Validation} \label{section:ExperimentalValidation}

In this section we present the performance of $\treeDecompHalin$. 
We compare it against the available implementations for computing tree decompositions presented in the PACE 2017~\citep{PACE2017} and in \texttt{libtw}~\citep{libtw}. 
Additionally, the correctness of the tree decomposition produced by \treeDecompHalin{} was verified using the validator provided in PACE 2017. 

In particular we compare against the two best exact algorithms from PACE 2017, referred to as E-Larisch and E-Tamaki and with the two best heuristics referred to as H-Tamaki and H-Strasser. 
For \texttt{libtw} we used the recommended approach based on the results reported by the authors. 
Specifically, a lower and an upper bound for the treewidth are first computed using the Least-C variant of Maximum Minimum Degree and the Greedy-FillIn heuristic, respectively.  
If the two values match, the tree decomposition is retrieved; otherwise, the QuickBB algorithm is used.  
We refer to this procedure as libtw(1).  
In practice we observed that the lower and upper bounds always matched. 
Therefore, we conducted a separate test that directly computes the tree decomposition using only the QuickBB algorithm which we refer to as libtw(2). 

We performed the experiments on a set of instances generated as follows. 
Let $N$ be a uniformly sampled integer within a given range. 
Create a graph with $N$ vertices, and for each $i \in [2, N]$ add an edge between vertex $i$ and a uniformly selected vertex from $[1, i-1]$. 
For each vertex of degree two create an additional vertex and connect it to it. 
Perform a depth-first search starting from vertex $1$, the order in which the leaves are visited defines the leaf cycle, so we add an edge between each pair of consecutive leaves. 
Finally, relabel the vertices using a random permutation. 
The correctness of this generator follows the same ideas presented in Section~\ref{sec:AdjacencyOrder}. 
We generated $40$ instances divided in four groups: Small, Medium, Large and Giant. 
The value of $N$ was sampled uniformly from the ranges $[10^2, 10^3-1], [10^3, 10^4-1], [10^4, 10^5-1]$ and $[10^5, 10^6-1]$ respectively. 
Note however that the final number of vertices may slightly exceed the upper bound due to the additional vertices created during the process. 

Recall that $\treeDecompHalin$ decomposes the Halin graph into the skeleton $\skeleton$ and the leaf cycle $C$ in linear time using the set of reductions presented in~\cite{Eppstein:2016}. 
The author provided an implementation in Python, a slow interpreted language in comparison to C/C++, and noted that time measurements would not serve as a reliable indicator of the algorithm's efficiency. 
We implemented the Halin recognition procedure in C++, including the construction of a valid decomposition into $\skeleton$ and $C$ in the positive case. 
Table~\ref{Table:HalinRecognitionTimes} presents the average runtime of both implementations across all instance groups. 
As expected the C++ implementation offers a clear performance advantage and is therefore used as the preprocessor in our experiments. 

All the experiments were executed on an Intel Xeon E5-2620 v2 2.1 GHz processor with 32 GB of RAM running on a Linux system. 
The implementation of \treeDecompHalin{}, the C++ version of Eppstein’s Halin recognition and decomposition and the set of instances used are freely available~\footnote{Since the paper is currently under submission, the source code is not yet publicly released. It is available upon request. A permanent link will be provided upon publication.}.

\begin{table}[t]
\centering
\begin{tabular}{ccc}\hline
       & Python    & C++     \\ \hline
Small  & 0.0749  &   0.0022\\
Medium & 0.1975  &   0.0090\\
Large  & 1.4396  &   0.0905\\
Giant  & 16.5318 &   1.5271\\ \hline
\end{tabular}
\caption{Average execution time of the Halin recognition and decomposition in Python and C++.}
\label{Table:HalinRecognitionTimes}
\end{table}

Table~\ref{Table:RuntimeExact} reports the average execution time of the exact algorithms for each group of instances. 
For all algorithms, the execution time was limited to a maximum of 24 hours. 
E-Larisch and E-Tamaki were unable to solve any instance in the Large and Giant groups due to segmentation faults caused by excessive memory usage. 
For the same reason, E-Tamaki was able to solve seven out of the ten Medium instances. 
Both libtw(1) and libtw(2) successfully solved eight Large instances. 
In the Giant group, libtw(1) solved three instances while libtw(2) solved four. 
Across all groups, $\treeDecompHalin$ was significantly faster.

\begin{table}[t]
\centering
\begin{tabular}{cccccc}\hline
       & $\treeDecompHalin$   & E-Larisch & E-Tamaki & libtw(1)  & libtw(2)  \\ \hline
Small  &  0.0066  & 0.0296  & 1.3460 & 0.2216 &  0.3225\\
Medium &  0.0269  & 0.5547  & 286.9072* & 4.6300 & 5.1874\\
Large  &  0.2896  & -       & -    &    581.8704* & 739.8171*\\
Giant  &  2.7925 & -       & -      &  8645.5772* & 49171.5802*\\ \hline
\end{tabular}
\caption{Average execution time of the exact algorithms. *The algorithm did not solved all the instances and the average was taken from the ones that solved.}
\label{Table:RuntimeExact}
\end{table}

Table~\ref{Table:RuntimeHeuristic} shows the performance of the heuristic algorithms. 
These algorithms expect a SIGTERM signal to terminate, at which point the best decomposition found is returned as output. 
The heuristics were run for $100$ times the execution time of $\treeDecompHalin$ (Table~\ref{Table:RuntimeExact}), after which they were terminated by sending a SIGTERM signal. 
The average width of the decompositions found is reported. 
Surprisingly, H-Strasser performed significantly better than H-Tamaki, the winner of PACE 2017. 
Note that, even with the huge increase of the stopping criterion, none of the heuristic methods could solve the medium, large or giant instances to optimality. 
Thus, \textsc{H-Td} is not only a proven exact method, but also demonstrates practical superiority over existing heuristic approaches on this class of graphs. 

\begin{table}[t]
\centering
\begin{tabular}{lcc} \hline
        & H-Tamaki    & H-Strasser \\ \hline
Small   & 581.3     & 3        \\
Medium  & 6041.3     & 3.3      \\
Large   & 44100.8  & 11.6     \\
Giant & 500761.5 & 31.7    \\ \hline
\end{tabular}
\caption{Average treewidth achieved by the heuristic algorithms when run for $100$ times the execution time of \treeDecompHalin{}.}
\label{Table:RuntimeHeuristic}
\end{table}

\section{Concluding remarks}
\label{Section:Concluding}

In this paper we proposed \textsc{H-Td}, a linear-time algorithm that computes an optimal-width tree decomposition for Halin graphs.
Unlike reduction-based approaches for partial $k$-trees, our method exploits the structural properties of Halin graphs and constructs the decomposition directly, without reducing the input graph.
This distinction is significant as the reduction rules for $k = 4$ and beyond become increasingly complex and impractical. 
Experimental validation against both exact and heuristic algorithms from PACE 2017 and \texttt{libtw} demonstrates that \textsc{H-Td} is not only theoretically efficient but also practical. 
We believe that this work highlights the potential of structural approaches and may inspire further development of optimal and efficient algorithms for other graph classes.

\bibliographystyle{unsrtnat}
\bibliography{references}  

\appendix

\section{How to compute Proposition~\ref{TEO:ORDEN} in linear time}
\label{appendix:LinearTeo1}
The root $r$ can be selected by iteratively searching a vertex with degree at least three in $O(|V_T|)$ time. To find a suitable $v$, we leverage the fact that $r$ has at least three distinct children. We can run a DFS starting from one of these children to explore and color its subtree. 
Given that $E_C$ forms a cycle, we can iterate over its edges. 
In doing so, we are guaranteed to find an adjacent pair of leaves $u$ and $w$ such that $u$ is colored and $w$ is not. 
By construction, the lowest common ancestor of $u$ and $w$ is $r$, allowing us to set $v = u$. 
As this process consists of a DFS and a traversal of $E_C$, it runs in $O(|V_T|)$ time. 

Note that once $r$, $v$ and $w$ are selected, we can compute the $x_u$ values for each leaf $u$ by traversing $E_C$. 
Recall that the $x$ value for an internal node $y$ is defined as the minimum $x$ value of a leaf inside the subtree of $y$. 
Thus, we can iterate the leaves in increasing order by their $x$ value and climbing toward the root, setting the $x$ value 
of each ancestor to its own $x$ value, stopping when an ancestor already has an assigned $x$ value. 
Since the $x$ value of each vertex is set exactly once, the overall procedure runs in linear amortized time.

Using the same principle, we can determine the adjacency order by first clearing the adjacency lists while preserving the parent-child relationships. 
Then, by iterating through the leaves in increasing order of their $x$ values and climbing toward the root, for each current vertex $y$, if the edge between $y$ and $p_y$ has not been created, we insert $y$ at the end of the adjacency list of its parent $p_y$ creating this edge; otherwise we stop climbing.

\section{$T(w, a, b)$ induces a tree decomposition when \treeDecompAHCall{t}{w}{a}{b} does not invoke any recursive call}
\label{appendix:BaseCase}
Note that if $a < b$, at least one recursive call will be invoked in \ref{AlmostHalin1:Step3}. 
Thus, it must hold that $a = b$. 
In this case, only \ref{AlmostHalin1:Step1} applies, as we ultimately proceed to \ref{AlmostHalin1:Step4}, where $T(w, a, b)$ is returned. 
To avoid invoking a recursive call, we must also ensure that $h(\child{w}{a}) = 0$.  
This condition holds only if $\child{w}{a}$ is a leaf, which implies that $\child{w}{a} = \leaf{ \child{w}{a} }{1}$. 
Therefore, the Almost Halin graphs for which \treeDecompAHCall{t}{w}{a}{b} does not invoke any recursive call are of the form illustrated in Figure~\ref{fig:AppendixB}.
Moreover, the tree decomposition returned consists of a single vertex with its bag equal to the set of representatives of the corresponding Almost Halin graph.
In this case, the decomposition conditions are trivially satisfied.


\begin{figure}[htbp]
    \centering
    \begin{subfigure}[b]{0.35\textwidth}
        \centering
        \includegraphics[width=\textwidth]{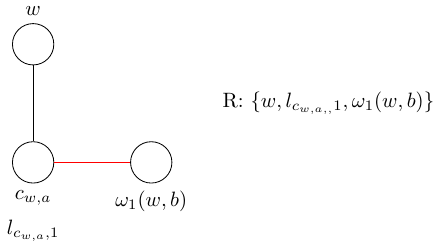}
        \caption{The case of an Almost Halin type $1$}
        \label{fig:imageA}
    \end{subfigure}
    \hfill
    \begin{subfigure}[b]{0.3\textwidth}
        \centering
        \includegraphics[width=\textwidth]{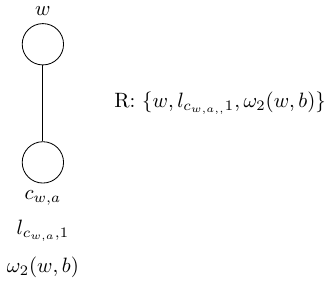}
        \caption{The case of an Almost Halin type $2$}
        \label{fig:imageB}
    \end{subfigure}
    \caption{Cases in which \treeDecompAHCall{t}{w}{a}{b} does not invoke any recursive call}
    \label{fig:AppendixB}
\end{figure}






\end{document}